\documentclass[12pt,letterpaper]{article}

\usepackage{fixltx2e}
\usepackage{textcomp}
\usepackage{fullpage}
\usepackage{amsfonts}
\usepackage{verbatim}
\usepackage[english]{babel}
\usepackage{pifont}
\usepackage{color}
\usepackage{setspace}
\usepackage{lscape}
\usepackage{indentfirst}
\usepackage[normalem]{ulem}
\usepackage{booktabs}
\usepackage{natbib}
\usepackage{float}
\usepackage{latexsym}
\usepackage{url}
\usepackage{hyperref}
\usepackage{epsfig}
\usepackage{graphicx}
\usepackage{amssymb}
\usepackage{amsmath}
\usepackage{bm}
\usepackage{array}
\usepackage{mhchem}
\usepackage{ifthen}
\usepackage{caption}
\usepackage{hyperref}
\usepackage{amsthm}
\usepackage{amstext}

\newtheorem{theorem}{Theorem}
\newtheorem{lemma}[theorem]{Lemma}

\newcommand{\PP}{{\mathbb P}}
\newcommand{\EE}{{\mathbb E}}

\raggedright
\renewcommand{\section}[1]{%
\bigskip
\begin{center}
\begin{Large}
\normalfont\scshape #1
\medskip
\end{Large}
\end{center}}

\renewcommand{\subsection}[1]{%
\bigskip
\begin{center}
\begin{large}
\normalfont\itshape #1
\end{large}
\end{center}}

\renewcommand{\subsubsection}[1]{%
\vspace{2ex}
\noindent
\textit{#1.}---}

\renewcommand{\tableofcontents}{}

\bibpunct{(}{)}{;}{a}{}{,}  

\begin{document}
\begin{flushright}
Version dated: \today
\end{flushright}
\bigskip
\noindent BRANCH LENGTHS ON TREES

\bigskip
\medskip
\begin{center}

\noindent{\Large \bf Branch Lengths on Birth-Death Trees and the Expected Loss of Phylogenetic Diversity}
\bigskip



\noindent {\normalsize \sc Arne Mooers$^1$, Olivier Gascuel$^2$, Tanja Stadler$^3$, Heyang Li$^4$, and Mike Steel$^4$}\\
\noindent {\small \it 
$^1$IRMACS, Simon Fraser University, Burnaby, BC, Canada V5A 1S6;\\
$^2$M\'{e}thodes et Algorithmes pour la Bioinformatique, LIRMM, CNRS - Universit\'{e} de Montpellier, 34095 Montpellier, France;\\
$^3$ETH Z{\"u}rich, Institut f{\"u}r Integrative Biologie, UniversitŠtstrasse 16, 8092 Z{\"u}rich, Switzerland;\\
$^4$Allan Wilson Centre for Molecular Ecology and Evolution, Biomathematics Research Centre, University of Canterbury, Christchurch, 8140, New Zealand;}\\
\end{center}
\medskip
\noindent{\bf Corresponding author:} Arne Mooers, IRMACS, Simon Fraser University, 8888 University Drive, Burnaby, BC, Canada V5A 1S6; E-mail: amooers@sfu.ca.\\


\vspace{1in}

\subsubsection{Abstract} 

Diversification is nested, and early models suggested this could lead to a great deal of evolutionary redundancy in the Tree of Life. This result is based on a particular set of branch lengths produced by the common coalescent, where pendant branches leading to tips can be very short compared to branches  deeper in the tree.  Here, we analyze alternative and more realistic Yule and birth-death models.   We show how censoring at the present both makes average branches one half what we might expect and makes pendant and interior branches roughly equal in length.  Although dependent on whether we condition on the size of the tree, its age, or both, these results hold both for the Yule model and for birth-death models with moderate extinction.  Importantly, the rough equivalency in interior and exterior branch lengths means the loss of evolutionary history with loss of species can be roughly linear.  Under these models, the Tree of Life may offer limited redundancy in the face of ongoing species loss. 

\noindent (Keywords: Phylogenetic tree, Yule process, extinction, phylogenetic diversity )\\


\vspace{1.5in}

In a well-cited paper,  \cite{nee97} state that ``80\% of the underlying tree of life can survive even when approximately 95\% of species are lost.'' This quote has percolated through the literature (see, e.g. \citep{erwin08, purvis08, roy09, santos10, vamosi08}). This high level of phylogenetic redundancy is due to Nee and May using coalescent-type models of tree shape, where pendant edges are expected to be much shorter than interior edges. Here, we test the robustness of this result by building on recent algebraic results from \cite{steel10} to derive the expected branch lengths on phylogenies produced under alternative Yule and birth-death models of diversification. We highlight three findings: (i) the average length of  branches in pure-birth (Yule) trees is roughly one half our naive expectation; (ii) the expected length of the interior branches and those leading to species are the same or nearly so, and this means that (iii) the relationship between the loss of species to extinction and the loss of phylogenetic diversity \citep{faith92} can be much more precipitous than that quoted above \citep{nee97}. All three findings hold for birth-death trees with low to moderate relative extinction rates.

For much of what follows, we will consider a pure-birth Yule tree with diversification rate $\lambda$. We note that inferred phylogenetic trees are often more imbalanced than Yule trees \citep{mooers97}, but currently, no biological model captures this empirical distribution. More importantly for what follows, the Yule process produces a distribution of splitting events on the tree from past to present that is intermediate between that expected under an adaptive radiation \citep{gavrilets05, rabosky08}, where splits are concentrated nearer the root,  that expected under long-term equilibrium models of diversification \citep{hey92, hubbell01}, where splits are concentrated nearer the present. Our main motivation for focusing on this model is that trees sampled from the literature tend to have splitting times concentrated nearer the root \citep{mcpeek08, morlon10}, making the Yule model a conservative model when measuring phylogenetic redundancy.

We refer to branches that lead to the tips of  a tree as pendant edges (with expected average length $p_n$, where $n$ is the number of tips) and branches found deeper within the tree as interior edges (with expected average length $i_n$). The term `expected average length' clarifies that two random processes are at work -- the production of a Yule tree and the selection of an edge from that tree. The expected phylogenetic diversity of such a tree is  the sum of the expected pendant and interior edge lengths, i.e. $L_n$ = $np_n$ +($n$-2)$i_n$. We will assume throughout that the tree starts as an initial bifurcation, such that at some time $t$ in the past it has two lineages each of length 0 (as in \cite{nee01}).  After time $t$ from the initial bifurcation, we produce a binary tree with $n$ tips (as in \cite{nee01, yang97}), and several properties of this process have been well-studied by these and other authors. In particular, the expected number of tips in the tree is $2e^{\lambda t}$.

 Given rate $\lambda$, the time that a given lineage persists until it splits on a Yule tree has an exponential distribution with a mean of $\frac{1}{\lambda}$.  This motivates our naive expectation that the expected average edge length on such a tree would also be $\frac{1}{\lambda}$.  We first present a simple proof that the expected average edge length in a Yule tree is actually $\frac{1}{2\lambda}$.  This provides an underlying intuition that is absent from the purely algebraic proof of \cite{steel10}. We then summarise and extend some  results from \cite{steel10}  to describe how the relative lengths of pendant and interior edges are affected by (i) conditioning on, (ii) estimating, or (iii) not knowing, three related quantities: $n$, the number of tips of the tree; $t$, the depth of the tree; and $\lambda$, the diversification rate. We then further extend our results to birth-death trees, and finally revisit the provocative question: at what rate do we lose phylogenetic diversity as we lose species on a tree?

\section{Expected length of a branch on a Yule tree sampled at the present}

Let us assume that we observe a Yule tree at the moment that it has grown to $n$ + 1 tips ($n$ =  4 in Fig.  1). We do not condition on its depth ($t$).  We can designate the edge that has just split as an interior edge, and disregard the two zero-length branches that have just arisen.  Doing so designates an equal number ($n-1$) of interior and pendant edges on this tree.  One might think of this Yule tree as one that has been `cut at' (or conditional on)  the observation of $n+1$ tips. Intuitively, even though the expected length of an edge on an uncensored tree would be $\frac{1}{\lambda}$, the designated pendant edges will be shorter due to this conditioning.  However,  interior branches are also affected by this censoring:  particularly long interior branches would stretch to the present, and so would be pendant edges.  This means that the expected lengths of interior edges are also shorter than $\frac{1}{\lambda}$.

\begin{theorem}
\label{gascuel}
In a Yule tree, at the latest speciation event, the expected length of a randomly drawn edge is $\frac{1}{2\lambda}$.
\end{theorem}

{\em Proof:}   Consider the late sampling scenario described in the preceding paragraph, and let the $n - 1$ remaining pendant edges each grow under the Yule process until they also split, disregarding all the new infinitesimal edges that result.  Each of these grown pendant edges has an expected length $g_n$ and is made up of two segments - its expected length before the tree had $n+1$ edges (= $p_b$), and its expected length as it continued to grow after the tree had $n+1$ tips  (= $p_a$), such that $g_n$ = $p_b$ + $p_a$.  Importantly, given the memoryless nature of exponential processes, the length of any pendant edge segment observed from the time that $n$ + 1 tips are produced (the dashed lines in Fig. 1)  is drawn from one common exponential distribution, with the same parameter $\lambda$. Also, $p_n$ on the censored tree = $p_b$ on the uncensored tree.

Given an equal number of interior and pendant edges on this uncensored tree, we can write an expression for the expected length (call it $\EE[L])$ of any randomly drawn edge on this tree as:

\begin {equation}
\label{one}
\EE[L]  = \frac{1}{2}\cdot i_n +   \frac{1}{2}\cdot (p_b +p_a) = \frac{1}{2}\cdot i_n + \frac{1}{2}\cdot (p_b +\frac{1}{\lambda}).
\end {equation}

Any single lineage has $ \EE[L] = \frac{1}{\lambda}$, and so we can substitute this for $\EE[L]$ to obtain:

\begin {equation}
\label{two}
\frac{1}{2}\cdot i_n +   \frac{1}{2}\cdot p_b  = \frac{1}{2\lambda},
\end {equation}
because  $p_n$ = $p_b$. The left member in equation (\ref{two}) is the expected length of a randomly drawn edge in the censored Yule tree, which completes the proof.

\begin{figure}[ht]
\begin{center}
\resizebox{11cm}{!}{
\includegraphics{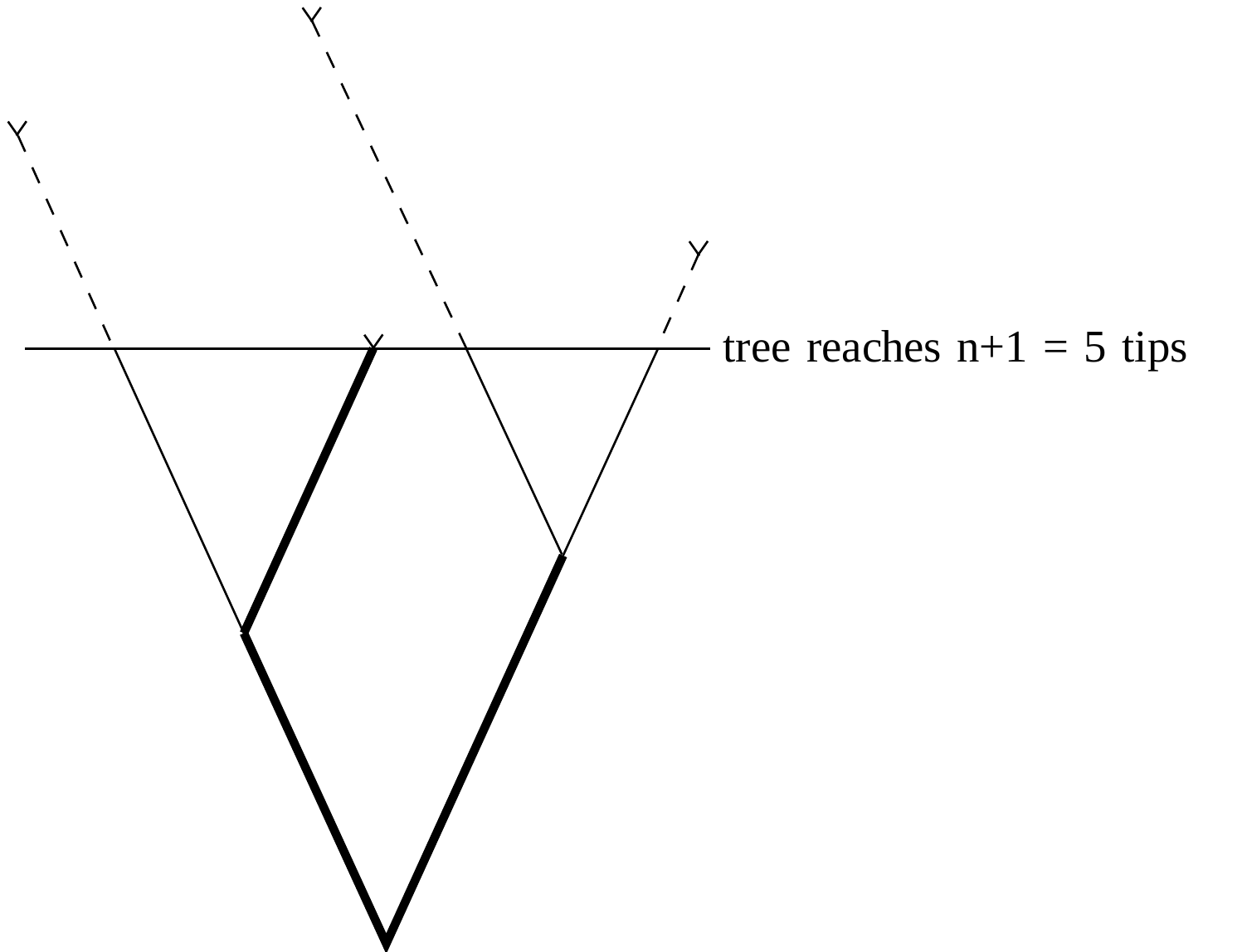}
}
\caption{Growing a Yule tree, to illustrate the proof of Theorem~\ref{gascuel}. The horizontal line is the observation time, when $n+1$ tips first appear. Below this line is the censored tree whose edge lengths we are modelling.  The uncensored tree has each pendant edge continuing to lengthen till it speciates in turn. The thick lines denote interior branches, the thin lines are pendant edges on the censored tree, and the dashed lines are the segments that accrue to produce the uncensored tree.
}
\label{fig1}
\end{center}
\end{figure}

This proof does not say anything about the relative lengths of internal vs. pendant edges per se - it might be that internal edges are still much longer than pendant ones on Yule trees that we observe at a single time slice, and it may be that the result hinges on observing the tree at exactly the moment that a speciation event occurs.  We turn to these issues now.

\section{Expected pendant vs. interior edge lengths as function (only) of $n$}

In the above construction of the Yule tree we made the convention that the edge that has just split is an interior edge of the resulting tree. However, we could have alternatively classified it as a pendant edge. In that case we have $n$ pendant edges and $n-2$ interior edges, and one can again consider the expected average pendant and interior branch lengths, which we will denote as $i'_n$ and $p'_n$.  Steel and Mooers (2010) used a recursive argument to
establish the following exact result: For all $n \geq 3$, we have:
\begin{equation}
\label{niceeq}
i'_n = p'_n = \frac{1}{2\lambda}.
\end{equation}

This result tells us exactly how $\frac{1}{\lambda}$ is shared out between the two terms in Theorem~\ref{gascuel}.  Due to the memoryless nature of the exponential distribution, the pendant edge that was chosen to split at our observation time is random with respect to its length, and so we can express the lengths of the interior and pendant edges on the censored tree as:
$$i_n = \frac{1}{n-1} ((n-2)i'_n + p_n'), \mbox{ and } p_n = \frac{1}{n-1}(np'_n - p'_n).$$
 Eqn. (\ref{niceeq}) implies that, for all $n \geq 3$:
$$i_n = i'_n = \frac{1}{2\lambda} \mbox{ and } p_n = p'_n = \frac{1}{2\lambda}.$$
In particular, the terms $i_n$ and $p_n$ in Theorem~\ref{gascuel} are equal.  We note that Theorem~\ref{gascuel} is for a late sampling scenario, when we show up just when $n$ + 1 tips first appear.  However, if we only condition on $n$, but show up at a random time between the interval when $n$ and $n$+1 tips exist (i.e. if we 'show up' at the present to sample our tree), any pendant edge has the same expected average length as in the late sampling scenario.  This result  is analogous to the bus-stop problem: if buses arrive at a certain rate $b$ under an exponential process, if one shows up at a random time, the expected time since the last bus  is $b^{-1}$ rather than something less than that.  This property was formally proven for model trees \citep{ger}  and also used recently by \cite{hartmann10} in the context of sampling trees from evolutionary models.

\section{Expected pendant vs. interior edge lengths as functions of $t$ (alone or with $n$)}
\label{onlytsec}
The expected number of tips in  a Yule tree at time $t$ is given by $N(t) = 2e^{\lambda t}$, since each of the two initial lineages has a geometrically distributed
distribution, with a mean of $e^{\lambda t}$ (see e.g. \cite{nee94}, or \cite{bei}  (Example 6.10, pp. 193)).    
We now introduce $P$ as the sum of all pendant edges, $I$ as the sum of all interior edges, and, as in the introduction, $L$ as the total tree length, $L$=$P$+$I$.  These quantities, conditional on either $n$ or $t$ or both, should be noted, as they will be useful for many of the proofs that follow.  If we let $P(t)$ and $I(t)$ denote, respectively, the expected sum of the lengths of the pendant and interior edges of a Yule tree grown for time $t$ and let $L(t) = P(t)+I(t)$, then,  from \cite{steel10}, we have the following equalities:

\begin{equation}
\label{tonly}
L(t) = \frac{2}{\lambda}(e^{\lambda t} - 1); P(t) = \frac{1}{\lambda}(e^{\lambda t} - e^{-\lambda t}) \mbox{ {\rm and} } I(t) = \frac{1}{\lambda}(e^{\lambda t} + e^{-\lambda t}-2).
\end{equation}

Thus the ratio of the expected average lengths of the pendant and interior edges of a Yule tree of depth $t$
 converges to $1$ exponentially fast with increasing $t$.  $P(t)$ is slightly larger than than $I(t)$, but the difference becomes rapidly negligible.
In particular, the ratio $P(t)/L(t)$ converges quickly to $1/2$;  we will consider this ratio further when we allow for extinction.

Importantly, for most phylogenetic trees, both $n$ and $t$ will be known from the data.  Do the observations on edge lengths made above also hold when we condition on both $n$ and $t$? The expected total length of a Yule tree conditional on it having grown for time $t$ and having exactly $n$ tips at time $t$ is given by:
\begin{equation}
\label{totaleq}
L_n(t) = t\cdot \left(2+ \frac{n-2}{x} (1-y(x))\right),
\end{equation}
where $x=\lambda t$ and $y(x): = \frac{xe^{-x}}{1-e^{-x}}$, which is a function that decreases from $1$ towards $0$ as $x \geq 0$ grows
(for details, see Steel and Mooers (2010)).
Let $I_n(t)$ and $P_n(t)$ denote the expected sum of the interior and pendant edge lengths (respectively) of a Yule tree,  conditional on it having  grown for time $t$ and having exactly $n$ tips at time $t$.
Thus, $I_n(t) + P_n(t) = L_n(t)$ (given by Eqn. (\ref{totaleq})).

A proof of the following result is provided in the Appendix.

\begin{theorem}
 \label{CorExpYule}
The expected length of a randomly picked pendant edge in a Yule tree on $n$ extant species and of age $t$ is,
$$\frac{1}{n}P_n(t) = t\cdot \left( \frac{2 }{n(n-1)} + \frac{(n-2) \left[(n+5) -4 (1+n +2x) e^{-x}+ (3n-1 +2(n+1) x )e^{-2x} \right]}{ 2x n(n-1)(1- e^{-x})^2}\right),$$
where $x=\lambda t$.

In particular, if we  set  $\lambda$ to its maximum likelihood estimate, i.e. $\lambda_{ML}=\log(\frac {n}{2}) / t$ \citep{magallon01}, then the ratio $\hat{R}_n:=P_n(t)/L_n(t)$ of the  expected total length of the pendant edges to the expected total length of all edges in a Yule tree on $n$ extant species and age $t$  is independent of $t$ and is given by:
$$\hat{R}_n = \frac{  n^3-3n^2 -4n \log(n/2) +4n -4}{ 2  (n-1)(n-2)^2},$$
which tends to $1/2$ as $n \rightarrow \infty$.
\end{theorem}

Table 1 presents  $P_n(t$), $L_n(t)$ and their ratio $\hat{R}_n(t)$ (ie, $R_n(t)$ conditioned on  $\lambda$ taking its  maximum likelihood estimate)  for a range of tree sizes.  

\begin{table}
\begin{center}
  \begin{tabular}{|p{0.5in}|r|r|r|}
    \hline
    $n=$ & $P_n(t)$ & $L_n(t)$  & $\hat{R}_n(t)$\\
    \hline
    $4$ & $3.03296\cdot t$ & $2.8854\cdot t$ & $1.0511$\\
    $16$ & $3.8697\cdot t$ & $6.7326\cdot t$ & $0.5748$\\
    $64$ & $9.2373\cdot t$ & $17.8894\cdot t$ & $0.5163$\\
    $256$ & $26.3815\cdot t$ & $52.3492\cdot t$ & $0.5040$\\
    $1024$ & $82.0735\cdot t$ & $163.8260\cdot t$ & $0.5010$\\
    
    \hline
  \end{tabular}
   \caption{Sum of pendant  edges ($P_n(t)$), sum of all edges ($L_n(t)$) and their ratio ($\hat{R}_n(t)$)  for various tree sizes $n$ when both $n$ and $t$ are fixed and $\lambda$ is set to its maximum likelihood value. }\label{Ta:first}
\end{center}
\end{table}

\section{Extension to birth-death models}
Allowing for random extinction (as well as speciation) introduces additional complexity into the analyses presented above.   We first consider what happens if we condition just on $n$ (and adopt the assumption that the  time of origin of the initial linage is a parameter of the birth-death model).  To do this, we have to assume a prior distribution for the time of origin  when conditioning the trees to have $n$ extant species.  We make the common assumption that the first species originated at any time in the past with uniform probability \citep{AlPo2005}. This is also called an improper prior on $(0,\infty)$.  Conditioning the resulting tree to have $n$ extant species yields a proper distribution for the time of origin \citep{ger}.  Note that, under the Yule model where $\mu$ = 0, this scenario is equivalent to stopping the process just before the $n+1$-th speciation event \citep{hartmann10}, which is the setting we considered in the first two sections of this paper.  The following result generalizes those earlier findings to birth-death models (a proof is provided in the Appendix).  As usual, $\lambda$ is the per lineage speciation rate and $\mu$ is the per-lineage extinction rate.

\begin{theorem}
\label{super1}
The expected length of a pendant edge on a birth-death tree conditioned on n is, for $0<\mu<\lambda$,
\begin{equation}
\label{lamu}
\EE[p|n] =\frac{ \mu + (\lambda - \mu) \log(1-\mu / \lambda)}{\mu^2};
\end{equation}
for $\mu = \lambda$, we have:
$$\EE[p|n] = \frac{1}{\lambda};$$
and  for $\mu=0$, we have:
$$\EE[p|n] = \frac{1}{2\lambda}.$$
\end{theorem}

We can also obtain exact results for the lengths of the edges in a birth-death tree
if we condition (just) on time. In particular, we can provide extensions to equation (\ref{tonly})to allow for extinction.  We begin, as
usual, with two lineages of length 0.
Let $T^R(t)$ denote the tree that is spanned by those taxa that are extant at time $t$; $T^R(t)$ is therefore referred to as the `reconstructed' birth-death tree (the tree consisting of edges that survive to time $t$ while extinct lineages are pruned away) \citep{nee94,  ger}. If there are no taxa extant at time $t$, we say that $T^R(t)$ is {\em empty}.
Let $N^R(t)$ denote the expected number of tips in the reconstructed birth-death tree, given by the well-known formula:
$$N^R(t) = 2e^{(\lambda-\mu)t}, t\geq 0.$$
Note that although $N^R(t)$ tends to infinity as $t$ grows when $\lambda> \mu$, it is quite possible that the actual number of lineages at time $t$ is $0$, in which case $T^R(t)$ is empty.
Let $L^R(t)$ be the expected total length of the reconstructed birth-death tree, and let
$P^R(t)$ be the expected sum of the pendant branch lengths of this tree.
The proof of the following result is provided in the Appendix.

\begin{theorem}
\label{birthdeath}
Consider a birth-death tree with speciation rate $\lambda>0$ and extinction rate $\mu$ that starts from two lineages of length $0$.
Let $\rho = \frac{\lambda}{\mu}$,
$r = \lambda - \mu$ and let $f_\rho(s) = \frac{\rho e^{s} - 1}{(\rho-1)e^{s}}.$
Then, for $t \geq 0$:
\begin{itemize}
\item[{\bf (i)}]
$L^R(t)=\frac{2 e^{rt}}{\mu} \cdot( \ln f_\rho(rt)).$
\item[{\bf (ii)}]
$P^R(t)=\frac{2e^{rt}}{\mu}  \left (1- (\rho -1)\cdot\left[(\ln f_\rho(rt))+\frac{1}{\rho e^{r t}-1}\right]\right).$
\item[{\bf (iii)}] For $\rho>1$,  the limiting ratio $\tau_\rho: = \lim_{t \rightarrow \infty} \frac{P^R(t)}{L^R(t)}$ is given by:
$$\tau_\rho  = \frac{1}{\ln \left[\frac{\rho}{\rho-1} \right]}-\rho+1.$$
\end{itemize}
\end{theorem}

\begin{figure}[ht]
\begin{center}
\resizebox{11cm}{!}{
\includegraphics{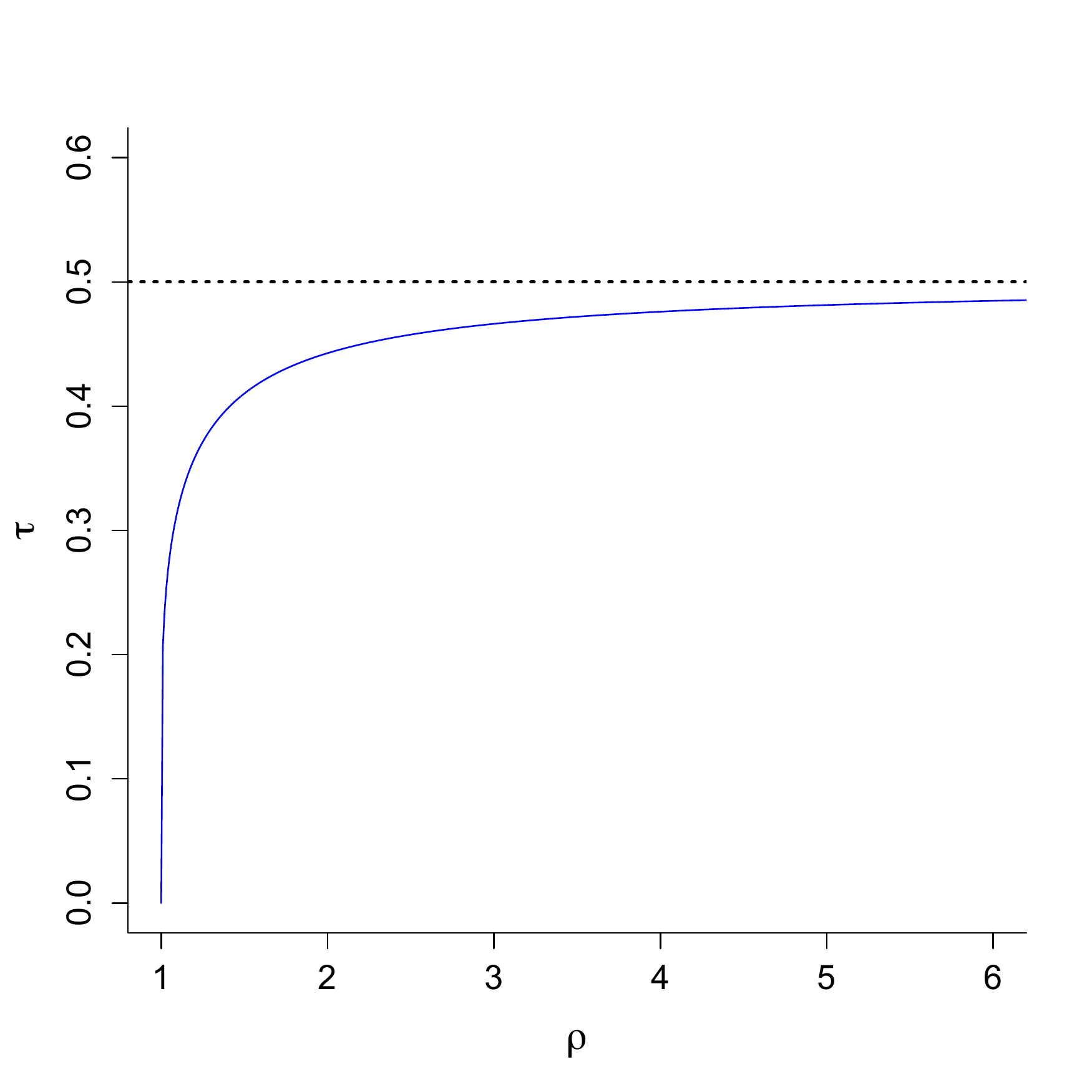}
}
\caption{Graph of $\tau_\rho$, which is the limiting ratio (for large $t$) of the sum of pendant edge lengths to the sum of all edge lengths in a birth-death tree, in which the speciation rate is $\rho>1$ times the extinction rate.}
\label{fig2}
\end{center}
\end{figure}

The function $\tau_\rho$ from part (iii) is shown in Fig. 2. Note that the 0.5 asymptote agrees with the ratio of $P^R(t)$ and $L^R(t)$ as in the pure-birth model as calculated earlier (i.e. $\tau_\rho \rightarrow \frac{1}{2}$ as
$\rho \rightarrow \infty$). Interestingly, the asymptote is reached fairly quickly on large trees. For example, from Fig. 2, we see that when the extinction rate is one-third of the extinction rate ($\rho$ = 3), then  $\tau_\rho$ = 0.47 and the expected pendant edge length is 87\%  the expected interior edge length.  Mild extinction in a uniform birth-death model does not produce particularly short pendant edges. At the other extreme, as the extinction rate approaches the speciation rate (so $r$ and $\rho$ converge to 0 and 1 respectively) $\tau_\rho$ can be easily shown to converge to $0$, as suggested by Fig. 2.
It is interesting to note that the expected sum of pendant edge lengths in the reconstructed tree at time $t$  (i.e. $P^R(t)$) divided by the expected number of extant taxa at time $t$ (i.e. $2e^{(\lambda - \mu)t}$) converges to the same expression as given in Eqn. (\ref{lamu}) as $t \rightarrow \infty$.

\section{Expected PD under simple Field-of-Bullets model for Yule trees}

The expected lengths of edges in a tree are directly relevant for quantifying the expected loss of `phylogenetic diversity' (PD) under simple models of extinction in which each tip is deleted with some fixed probability.  In these models, edges that are `deep' within the tree are more likely to contribute to the PD score of the surviving taxa than pendant edges of similar length, since they are more likely to have
at least one non-extinct taxon in the clade they support. This redundancy leads to the nonlinear decrease of PD as more species are removed from a tree \citep{nee97}.  However, the ratio of the lengths of pendant  to interior edges is also critical, as pendant edges will be the first to be deleted from the tree.   In this section, we analyse the expected PD score of a Yule tree under random taxon deletion. Note that there are {\em two} random processes at play here:  the Yule process that produces the tree, and then the extinction process that deletes taxa.

Consider then  a Yule tree that starts with a split into two lineages at time $0$ and is grown until time $t>0$. At that time, each tip is selected independently with probability $s$, and the remaining tips
are deleted (pruned). Thus $s$ is the `survival probability' of a taxon.   Let $\psi_t(s)$ be the $PD$ of the resulted pruned tree, and let
$\pi_t(s)= \EE[\psi_t(s)]$, where $\EE[.]$ denotes expectation with respect to the random Yule tree and the random pruning operation.
Thus, $\pi_t(1)$ is the expected $PD$ of the (entire) Yule tree, namely $L(t) =\frac{2}{\lambda}(e^{\lambda t} -1 )$ (Eqn.  \ref{tonly}).
 For $s<1$, $\pi_t(s)$ is the expected $PD$ one obtains by generating a
Yule tree until time $t$ and then applying a field-of-bullets pruning with survival probability $s$ for each tip.  The proof of the following result is provided in the Appendix.

 \begin{theorem}
 \label{best}
$$\pi_t(s)= \frac{2s}{(1-s)\lambda}e^{\lambda t} \cdot \left[ - \log\left(s + (1-s)e^{-\lambda t}\right)\right].$$
The ratio  $\pi_t(s)/\pi_t(1)$ of the expected PD in the pruned tree to the expected PD of the total tree therefore converges (quickly) with $t$ to the limit:
$$\pi(s):= \frac{-s\log(s)}{1-s}$$
\end{theorem}

Theorem~\ref{best} implies  that $\pi_t(s) \geq s \cdot  \pi_t(1)$ for all $t>0$.
Moreover, the limiting ratio  $\pi(s)$ is a continuous and concave, positive function that approaches $0$ as $s \rightarrow 0$ and approaches  $1$ as $s \rightarrow 1$ (see Fig. 3). For $s=0.5, \pi(s) =\log(2) = 0.69.$
The slope function $\pi'(s)$ approaches infinity as $s$ approaches $0$ from above and $\pi'(s)$ approaches $\frac{1}{2}$ as $s$ approaches $1$ from below. This latter result can be seen by considering that pendant edges are the first to be lost from a tree undergoing extinction; under the Yule model, the sum of the pendant edges constitutes 0.5 of the total PD (Theorem~\ref{tonly}).

\begin{figure}[ht]
\begin{center}
\resizebox{11cm}{!}{
\includegraphics{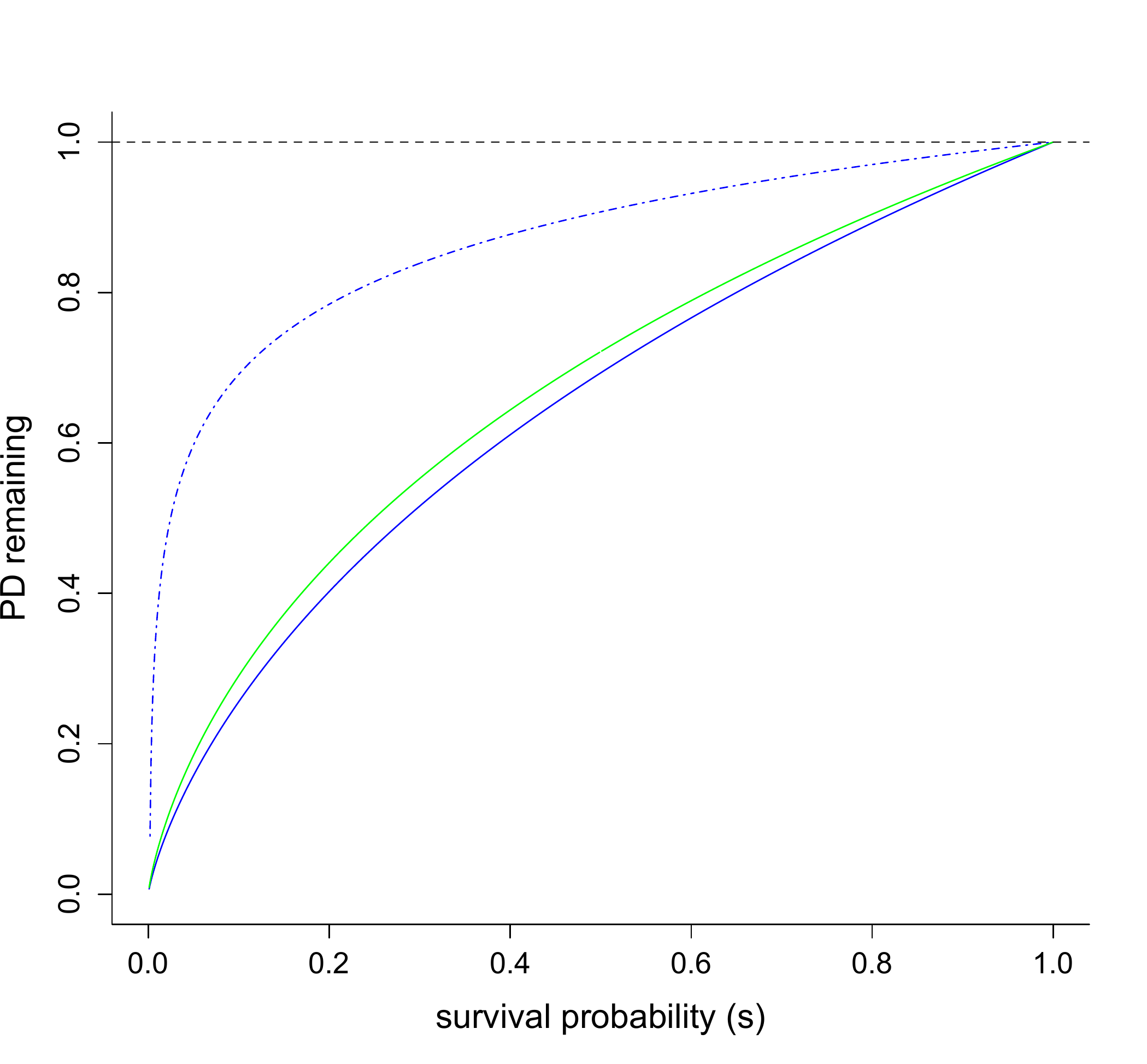}
}
\caption{Lower solid line shows the proportion of PD remaining when random extinction occurs with probability 1- $s$ on a Yule tree (from $\pi(s)$ from Theorem~\ref{best}). The dotted line shows the same quantity for the coalescent-style tree used by \citep{nee97}, for $n$=1000.  The curve in between these two shows the same quantity but on  a birth-death tree with $\mu = 0.5\lambda$, as described by Eqn. (\ref{mumu2}).}
\label{fig3}
\end{center}
\end{figure}

The high level of redundancy reported by  \cite{nee97} is due to their use of  coalescent-type models of tree shape with a constant population-size, where the pendant edges are expected to be much shorter than the interior edges.  More precisely, the ratio of the expected total length of the pendant edges to the expected total length of the interior edges converges to $0$ with increasing $n$, at a rate $1/\log(n)$, see e.g.  \citep{fuli} (Eqns. (10-12)).  An example of the relationship between $s$ and  the proportion of the tree remaining under Nee and May's model (for $n$ = 1000) is shown in Fig. 3.  

Under a Yule model, where interior and pendant edges have roughly the same expected length, the situation is quite different. If we take $s$=0.05, then $\pi(s)$ = 0.157. That is, in a large tree, if we lose 95\% of species (randomly) then we would expect to {\it lose} more than 84\% of the tree.  This lower level of redundancy is also more in line with statistical \citep{morlon11} and empirical estimates of tree loss under extinction regimes \citep{voneuler01,purvis00, vamosi08}, where tree shape and non-random extinction interact (see also \citep{heard00, nee05}).

Similar results hold with birth-death trees under mild extinction (see  Fig. 2), where the sum of the pendant edges constitutes $\tau_\rho$ of the total PD.   In particular, for
$\lambda>\mu >0$, the 
second formula presented in Theorem~\ref{best} can be modified as follows (see Appendix):
\begin{equation}
\label{mumu2}
\pi(s)=\frac{s}{(a-s)}\cdot\log \left(\frac{s}{a}\right)\cdot \frac{1-a}{\log(a)},
\end{equation}
 where $a= 1-\mu/\lambda$.
 
Fig. 3 exhibits  an example curve $\pi(s)$ on a birth-death tree constructed with $\mu$ = 0.5$\lambda$.

We note that this modified formula for $\mu>0$ should be used with care for larger values of $\mu$ for two reasons. Firstly, birth-death trees are increasingly likely to die out as $\mu$ approaches $\lambda$ and so an asymptotic  ratio of expected values such as $\pi(s)$ may be a poor estimate of expected PD loss in such situations. Note in particular that in the limit as $\mu/\lambda \rightarrow 1$, we have $\pi(s)=1$ for all $s>0$. This of course does not mean that if 99.9\% of the taxa are eliminated, then we would still expect to  retain 100\% of the phylogenetic diversity!

The second reason for caution is more empirically based.  In the extreme (critical) case where $\mu=\lambda$ then, as we have noted already, if we condition on a tree having $n$ extant leaves (assuming a uniform prior distribution for the time of the origin of the tree, as in  \citep{AlPo2005}), then the expected distribution of branch lengths in this tree would be precisely that given by the coalescent process \citep{gern08} that was used in the analysis by  \cite{nee97} . The problem now is that typical species-level phylogenetic trees look very different from such constant-size coalescent-shaped trees.  \cite{hey92}, using a sample of only eight trees,  was the first to point out that the coalescent model produced unreasonably short pendant edges (see also \cite{morlon11},  while McPeek's recent compilation \cite{mcpeek08} of 245 fairly-well sampled chordate, arthropod, mollusk, and magnoliophyte phylogenies, showed that these trees tended to have a branch length distribution in the opposite direction to the coalescent, with edges near the leaves tending to be, on average, slightly $longer$ than expected under the Yule model.   McPeek used the gamma statistic from \citep{pybus00} to describe the distribution of branch lengths as one moves from the root of the tree to the tips, and found that the majority of trees had negative gamma values, rather than having them centered on $0$ as expected under the Yule model \citep{pybus00} and the positive values expected under the coalescent \citep{pybus02}. Indeed, we show in the appendix that the expected value of gamma for a coalescent tree of increases indefinitely at a rate of  $\sqrt{3n}$.

\cite{morlon10} applied a coalescent framework that allows for incomplete taxon sampling to an overlapping set of 289 trees and found that the majority of trees ($>80\%$) had splitting times that were either consistent with the Yule model or  concentrated nearer the root.  Though nonrandom sampling may be a concern \citep{cusimano10}, the observation that most nearly-complete phylogenetic trees have gamma values close to zero (or negative), as well as the explicit test of the Yule model by \cite{morlon10}   suggest that our use of this model in analyzing expected loss of PD may be conservative.

\section{Conclusion}

Although the Yule model of diversification is nearly 100 years old, it still holds some surprises.  The fact that real  trees are conditioned on $t$ and that we show up at some random time after $n$ tips have been produced  leads to the observation that average pendant edge lengths (species ages) and internal edge lengths (those that anchor higher clades) are expected to be nearly equal under the Yule model.  Although all edges are not expected to be the same length  -- for instance the two edges incident to the root are longer than others (results not shown) -- this conditioning  also makes randomly selected edge lengths one half of  the naive expectation. These observations may be useful in informing prior distributions on edge lengths for tree inference.

Mild amounts of uniform extinction do not change these general observations.  Indeed, the `push of the past'  \citep{harvey94, phillimore08}, which describes the expectation that those groups which diversified faster than expected early on are more likely to be sampled in the present, would lead to internal edges being even shorter relative to pendant edges.  Non-uniform models, such as adaptive radiations where diversification actually slows down through time \citep{rabosky08, morlon10}, would do the same.  All these processes work against the redundancy inherent in the Tree of Life. We predict that this redundancy may not be as great as hoped for.  Of course, this prediction must await more complete, dated trees.

\section{Funding}
This work was funded by the Royal Society of New Zealand James Cook Fellowship and Marsden Funds (MS and LH),  the PhyloSpace project (ANR - Programme la 6\`{e}me Extinction, OG), and the Natural Sciences and Engineering Research Council of Canada Discovery Grants programme (AOM).

\section{Acknowledgements}
We thank our funders for support, and two anonymous reviewers, various audiences and the Associate Editor C\'{e}cile An\'{e} for very helpful comments.
 
 \newpage
 
\bibliographystyle{sysbio}
\bibliography{Arne_et_al}

\newpage

\section{Appendix: Proofs of Theorems}
\subsection{Proof of Theorem \ref{CorExpYule}}
\label{theoremtwoproof}

We can modify  the argument that leads to  the differential equation $\frac{dI(t)}{dt} =\lambda P(t)$ from \citep{steel10}  so as to take into account conditioning on $n$  as well as $t$ -- the analysis consists of calculating
quantities such as $\PP[X_t=n-1|X_{t+\delta}=n]$, where $X_t$ denotes the number of species present at time $t$, for which Eqn. (4) of \citep{nee01} is helpful.
In this way one can  derive the following sequence of first-order linear differential equations for $I_n=I_n(t)$:
\begin{equation}
\label{bigeq}
\frac{dI_n}{dt} + \frac{\lambda(n-2)}{1-e^{-\lambda t}}\cdot  I_n =  \frac{\lambda(n-2)}{1-e^{-\lambda t}}\cdot \left (I_{n-1}+\frac{1}{n-1}P_{n-1}\right).
\end{equation}

Notice that the term $P_{n-1} = P_{n-1}(t)$ on the right-hand side of (\ref{bigeq}) can be replaced by $L_{n-1}(t) - I_{n-1}(t)$ (with $L_{n-1}(t)$ given by (\ref{totaleq})).   Moreover, when $n=2$ we have the initial solution
$I_2(t) =0$ (and $P_2(t)=2t$) for all $t \geq 0$, and for each $n$ we have the boundary condition $I_n(t) =0$  at $t=0$.  

It can now be verified that the expression given in Theorem \ref{CorExpYule} for $P_n(t)$ satisfies this system of linear differential equations subject to the boundary condition, and so is the unique solution.  

For the second claim if we set  $\lambda$ to its maximum likelihood estimate, i.e. $\lambda_{ML}=\log(\frac {n}{2}) / t$, then,

{\footnotesize
\begin{eqnarray*}
P_n(t) &=& \frac{2 \log(n/2) }{\lambda (n-1)}+\\
& & \frac{(n-2) \left[(n+5) -4 (1+n +2\log(n/2)) e^{-\log(n/2)}+ (3n-1 +2(n+1) \log(n/2) )e^{-2 \log(n/2)} \right]}{ 2 \lambda (n-1)(1- e^{-\log(n/2)})^2}\\
&=& \frac{2 \log(n/2) }{\lambda (n-1)}+\frac{(n-2) \left[(n+5) -4 (1+n +2\log(n/2)) 2n^{-1}+ (3n-1 +2(n+1) \log(n/2) )  (n/2)^{-2}  \right]}{ 2 \lambda (n-1)(1-2/n)^2}\\
&=& t \frac{ n^3-3n^2 -4n\log(n/2) +4n -4}{ 2 \log(n/2) (n-1)(n-2)}\\
\end{eqnarray*}
}

The sum of all edge lengths is in expectation \citep{steel10},
$L_n(t)= t \frac{n-2}{\log(n/2)},$
and therefore the ratio $R_n$ is 
the expression given Theorem~\ref{CorExpYule}. From this expression it is easily seen that
 $\lim_{n\rightarrow \infty} \frac{P_n(t)}{L_n(t)}= 1/2$.
\hfill$\Box$

\bigskip

\subsection{Proof of Theorem~\ref{super1}}
\label{theoremthreeproof}

The probability $v(k)$ that a leaf is attached to the $k$th speciation event  in a tree on $n$ extant species under the Yule or birth-death model is, from  \cite{Sta}, given by: \begin{eqnarray}
v(k) = \frac{2k}{n(n-1)}. \label{pk}
\end{eqnarray}
\noindent
For $ 0 \leq \mu < \lambda$, let:
$$p_0(t):= \frac{(1-e^{-(\lambda-\mu)t})}{\lambda-\mu e^{-(\lambda-\mu)t}},\mbox{ and } 
p_1(t):= \frac{ (\lambda-\mu)^2 e^{-(\lambda-\mu)t}}{(\lambda - \mu e^{-(\lambda-\mu)t})^2}, $$
while for $\mu = \lambda$, let:
$$
p_0(t) := \frac{ t}{1+\lambda t}, \mbox{ and }  p_1(t) :=  \frac{1}{(1+\lambda t)^2}.$$
The probability that a lineage produces $0$ (resp. $1$) offspring after time $t$ is $\mu p_0(t)$ (resp. $p_1(t)$) \cite{Kendall1949}.
We first establish the following result:
\begin{lemma} \label{Thmpendn}
The length of a randomly picked pendant edge in a birth-death tree on $n$ extant species has probability density function 
$f_p(t|n) = 2 \lambda  p_1(t) (1-\lambda p_0(t)).$
\end{lemma}
\begin{proof}
For proving the lemma, we will use the probability density of the time of the $k$-th speciation event in a birth-death tree with $n$ extant species which is derived in \cite{ger}, and for $\mu<\lambda$, we get,
\begin{eqnarray}
f_{n,k}(t) = (k+1){n \choose k+1} \lambda^{n-k} (\lambda-\mu)^{k+2}
 e^{-(\lambda-\mu)(k+1)t} 
 \frac{\left( 1- e^{-(\lambda-\mu)t}\right)^{n-k-1}}{(\lambda-\mu  e^{-(\lambda-\mu)t})^{n+1}}. \label{fnks}
\end{eqnarray}
Using Equation (\ref{pk}) and (\ref{fnks}), we can write,
\begin{eqnarray*}
f_p(t|n) &=& \sum_{k=1}^{n-1} v(k) f_{n,k}(t)\\
&=& 2 \sum_{k=1}^{n-1} {n-2 \choose k-1}  \lambda^{n-k} (\lambda-\mu)^{k+2}
 e^{-(\lambda-\mu)(k+1)t} 
 \frac{\left( 1- e^{-(\lambda-\mu)t}\right)^{n-k-1}}{(\lambda-\mu  e^{-(\lambda-\mu)t})^{n+1}}\\
 &=& 2 \lambda^{n-1} (\lambda-\mu)^3  e^{-2(\lambda-\mu)t}  \frac{\left( 1- e^{-(\lambda-\mu)t}\right)^{n-2}}{(\lambda-\mu  e^{-(\lambda-\mu)t})^{n+1}}
  \sum_{k=1}^{n-1} {n-2 \choose k-1}  \left( \frac{(\lambda-\mu) e^{-(\lambda-\mu)t} }{\lambda(1- e^{-(\lambda-\mu)t})}   \right)^{k-1}\\
    &=&  2 \lambda (\lambda-\mu)^3 \frac{ e^{-2(\lambda-\mu)t}}{(\lambda - \mu e^{-(\lambda-\mu)t})^3}.
 \end{eqnarray*}
 For $\mu=\lambda$, we take the limit $\mu \rightarrow \lambda$ (using the property $e^{-\epsilon} \sim 1-\epsilon$), which establishes the lemma.
\end{proof}
Note that the length of a pendant edge is independent of $n$.
Theorem \ref{super1}  now  follows directly from Lemma \ref{Thmpendn} by evaluating $\int_0^\infty t f_p(t|n) dt$.
\hfill$\Box$

\bigskip

\subsection{Proof of Theorem~\ref{birthdeath}}
\label{theoremfourproof}

 The quantity  $\frac{1}{f_\rho(rt)} = \frac{r}{\lambda - \mu e^{-rt}}$ is the probability that a birth-death tree that starts with a single lineage at time $0$ has at least one extant lineage at time $t$ (Eqn. (2) of \citep{nee94}). 
Thus, by considering the first $\delta$ period of time in a birth-death tree that begins with a single lineage, the expected total sum $S(t)$ of branch lengths spanning the leaves
present at time $t$ satisfies the differential expression:

$$S(t+\delta) = 0 \cdot \mu\delta + 2S(t) \cdot \lambda \delta + (S(t)+\delta \frac{1}{f_\rho(rt)})\cdot (1-(\mu+\lambda) \delta) + O(\delta^2),$$

(by considering whether or not the lineage becomes extinct, speciates, or persists unchanged within this initial $\delta$ period).  
Since $L^R(t) = 2 S(t)$ this leads to the following differential equation:
\begin{equation}
\label{thomas1}
\frac{dL^R(t)}{dt} = rL^R(t) + 2/f_\rho(rt).
\end{equation}
Solving Eqn. (\ref{thomas1}) subject to $L^R(0)=0$, gives part (i) of the Theorem.
By considering the evolution of the tree from time $t$ to $t+\delta$ a straightforward dynamical argument leads to a second differential equation that links $L^R(t)$ to $P^R(t)$ :
\begin{equation}
\label{thomas2}
\frac{dL^R(t)}{dt} = N^R(t) - \mu P^R(t).
\end{equation}
Part (ii) follows by equating the right-hand sides of Eqns. (\ref{thomas1}) and (\ref{thomas2}) to express $P^R(t)$ in terms of quantities already determined.
For Part (iii), observe that $r>0$ and $f_\rho(rt) \rightarrow (\rho-1)/\rho$ as $t \rightarrow \infty$ and so, from parts (i), (ii), we have the asymptotic equivalences
$L^R(t)/2e^{rt} \sim \mu^{-1}\ln[\rho/(\rho-1)], P^R(t)/2e^{rt} \sim \mu^{-1}(1- (\rho-1) \ln[\rho/(\rho-1)])$. Taking the ratio of these quantities gives the result claimed.

\bigskip

\subsection{Proof of Theorem~\ref{best}}
\label{theoremfiveproof}

Let $\phi_t= \phi_t(s)$ be the analogue of $\psi_t(s)$ if we start the Yule tree with a single (rather than 2) lineages at time $t=0$;  thus,
\begin{equation}
\label{phieq}
\pi_t(s) = \EE[\psi_t(s)] =2 \EE[\phi_t(s)],
\end{equation}
 (the behaviour of $\phi$ is slightly easier to analyse than $\psi$).
Let $X_t$ denote the number of tips in the Yule tree (starting with a single lineage at time 0) at time $t$.  Consider $\phi_{t+\delta}$, for a small value $\delta>0$.  In the first $\delta$ period of time the initial lineage can either  (i) speciate (with probability $\lambda \delta + O(\delta^2)$)  or (ii) fail to speciate (with probability $1-\lambda\delta + O(\delta^2)$) and so  we have:
\begin{equation}
\label{caseseq}
\phi_{t+\delta}= \begin{cases}
\phi_t^1 + \phi_t^2 + O(\delta), &\text{ with probability } \lambda \delta + O(\delta^2);\\
\phi_t ^0+ Y_t,  &\text{ with probability } 1-\lambda \delta + O(\delta^2);
\end{cases}
\end{equation}
where
$$\EE[Y_t|X_{t+\delta}=n] = \delta \cdot (1- (1-s)^n),$$
and $\phi_t^0$, $\phi_t^1$ and $\phi_t^2$ are independent random variables having the same distribution as $\phi_t$ (the contribution of $\delta$ to the PD score of the tree applies precisely if at least one of the
tips at time $t+\delta$ is sampled, and this event, conditional on $X_{t+\delta}=n$, has probability $1- (1-s)^{n}$).
Now, $$\PP(X_{t+\delta}=n|X_\delta = 1) = \PP(X_t=n|X_0=1),$$ and it is a classic result that this latter probability has a geometric distribution with mean $e^{\lambda t}$  (see e.g.  \cite{bei},  Example 6.10, pp. 193) and so:
\begin{equation}
\label{expeq1}
\EE[Y_t] = \delta \cdot (1-\EE[(1-s)^{X_t}]) = \delta \cdot \left(1-\sum_{n \geq 1} (1-s)^n e^{-\lambda t}(1-e^{-\lambda t})^{n-1}\right) = \frac{\delta \cdot s}{s+qe^{-\lambda t}}
\end{equation}
where $q=1-s$.  Let $\pi'_t(s) : = \EE[\phi_t(s)]$.  Taking expectation of (\ref{caseseq})  (with respect to both the Yule tree and the random sampling process) and applying (\ref{expeq1})  leads to the
 following differential relationship for  $\pi'_t(s)$:
 $$\pi'_{t+\delta}(s) = 2\lambda\delta\cdot\pi'_t(s) + (1-\lambda\delta) \cdot \left( \pi'_t(s) + \frac{\delta\cdot s}{s+qe^{-\lambda t}}\right) + O(\delta^2).$$
 This leads to the following first-order, linear differential equation for $\pi_t'(s)$:
$$\frac{d\pi'_t(s)}{dt}-\lambda \pi'_t(s) = \frac{s}{s+qe^{-\lambda t}}. $$
Solving this equation gives $\pi'_t(s)$, and thereby the stated value for $\mu_t(s) = 2\pi'_t(s)$ (by (\ref{phieq})).

The modification of this result to give Eqn. (\ref{mumu2}) in the birth-death setting, with $0<\mu < \lambda$ following a similar case analysis (but allowing for the possibility of extinction) leads to the differential equation for $M_t(s) = \EE[\phi_t(s)]$:
\begin{equation}
\label{linearde}
\frac{dM_t(s)}{dt}=(\lambda-\mu)M_t(s)+\PP(\phi_t(s)\neq0).
\end{equation}
Now,  by Eqn. (1) of  \citep{yang97} (or see \citep{nee94}) we have:
\begin{equation}
\label{expeq}
 \PP(\phi_t(s) \neq 0) =\frac{as}{s-(s-a)e^{-(\lambda-\mu)t}},
\end{equation}
where $a=1-\mu/\lambda$.  
Now $\pi_t(s)$ lies between $2M_t(s)$ and $2M_t(s)-t$ (depending on whether we add the lengths of all the edges from the extant taxa to the root, or just the edges  from the extant taxa to their most recent common ancestor), from which Eqn. (\ref{mumu2}) follows by evaluating the limit  of the ratio $\pi_t(s)/\pi_t(1)$ as $t \rightarrow \infty$.

\hfill $\Box$.

\bigskip

\subsection{The expected value of  $gamma$  under the coalescent process}
\label{theoremsixproof}

Under a Yule (pure-birth) model, the gamma statistic has a standard normal distribution with mean 0, while under a coalescent model it is positive.   Under the coalescent model,  the original $\gamma$ statistic grows at the asymptotic rate of $\sqrt{n}$ as the number of tips $n$ grows.

{\bf Theorem 6.} 
\indent {\em For a coalescent tree with $n$ leaves, $\gamma/\sqrt{n}$ converges in probability to  $\sqrt{3}$ with increasing $n$.}

For a rooted binary tree with $n\geq 2$ leaves, let $g_2, g_3, \ldots, g_n$ be times between successive speciation events, 
measured from the root to the leaves, and let
$T_n = \sum_{j=2}^n jg_j$. 
From \cite{pybus00} we have $\gamma = \frac{X_n}{Y_n},$ where $X_n$ can be written in the form:
$$X_n = \frac{1}{n-2} \sum_{i=2}^{n}\alpha_i g_i, \mbox{ where } \alpha_i = i(n/2) - 2\binom{i}{2},$$
and
$$Y_n = T_n \sqrt{\frac{1}{12(n-2)}}.$$
Now,  under the coalescent, the random variables $g_2, \ldots, g_{n}$ are independently distributed, and with $g_j$ having an exponential distribution with mean $\frac{1}{\binom{j}{2}}$.   It follows that $\frac{T_n}{2 \log(n)}$ and $\frac{X_n}{\log(n)}$ have expected values that converge to $1$, and variances that converge to $0$ as $n \rightarrow \infty$,  and so
$\frac{T_n}{2 \log(n)}$ and $\frac{X_n}{\log(n)}$ each converge in probability to the constant  $1$ as $n \rightarrow \infty$.  Consequently, the ratio 
$X_n/T_n$ converges in probability to 1/2 as $n \rightarrow \infty$, and so $\gamma(n)/\sqrt{n} = \frac{X_n}{T_n}\cdot  \frac{\sqrt{12(n-2)}}{\sqrt{n}}$ converges in probability to 
$\sqrt{3}$, as claimed.

Finally, a more careful asymptotic analysis provides a closer approximation to $\gamma/\sqrt{n}$ by the formula $\sqrt{3}\cdot (1-\frac{2}{\log_e(n)+C})$ where
$C$ is Euler's constant (0.5772...), and simulations confirm this improved fit.

\hfill $\Box$

\newpage


\end{document}